\documentclass[11pt]{proposalnsf}
\usepackage[compact]{titlesec}
\usepackage[utf8]{inputenc}
\usepackage{amsopn}
\usepackage{csquotes}
\usepackage{color,soul}
\usepackage{amsfonts}
\usepackage{amssymb}
\usepackage{paralist}
\usepackage[export]{adjustbox}
\usepackage{wrapfig}
\usepackage{subcaption}
\usepackage[justification=centering]{caption}
\usepackage{titling}

\usepackage{blindtext}
\usepackage{graphicx}
\usepackage{color}
\usepackage{caption}
\usepackage{subcaption}
\usepackage{amssymb}
 \usepackage{amsthm}
 \usepackage{float}
 \usepackage{tikz}

\newtheorem{thm}{Theorem}[section]
\newtheorem{cor}[thm]{Corollary}
\newtheorem{lem}[thm]{Lemma}

\theoremstyle{definition}
\theoremstyle{remark}

\newtheorem{definition}{Definition}

\usepackage{amsmath}

\usepackage{color}

\graphicspath{{Images/}}

\DeclareMathOperator{\BHV}{BHV}

\begin{document}

\title{A combinatorial method for connecting BHV spaces representing different numbers of taxa}

\maketitle

%\noindent{\Large \bf A Combinatorial Method for Connecting BHV Spaces Representing Different Numbers of Taxa}
\begin{center}

% We don't use a special title page; the author information is entered 
% like any other text.

% FOOTNOTES: We don't allow them in the manuscript, except in
% tables. Don't include any footnotes in the text.

%\noindent {\normalsize \sc First Author$^1$, Second Author$^2$, and Third Author$^1$}\\
%\noindent {\small \it 
%$^1$Department, University, City, State, Zip Code, Country;\\
%$^2$Department, Institution, City, State, Zip Code, Country}\\
%\end{center}
%\medskip
%\noindent{\bf Corresponding author:} Name, Department, University, Address,
%City, State, Zip Code, Country; E-mail: corresponding.author@univ.edu.\\

\noindent {\normalsize \sc Yingying Ren$^5$, Sihan Zha$^6$, Jingwen Bi$^2$, Jos\'{e} A. Sanchez$^1$, Cara Monical$^1$,  Michelle Delcourt$^3$, Rosemary Guzman$^1$, and Ruth Davidson$^4$}\\
\noindent {\small \it 
$^1$Department of Mathematics, University of Illinois Urbana-Champaign, Urbana, Illinois, 61801, U.S.A;\\
$^2$Department of Engineering, Cornell University, New York, New York, 10044, U.S.A.;\\
$^3$School of Mathematics, University of Birmingham, Edbagston, Birmingham, B15 2TS, U.K.\\
$^4$Departments of Mathematics and Plant Biology, University of Illinois Urbana-Champaign, Urbana, Illinois, 61801 U.S.A.\\
$^5$Departments of Mathematics and Computer Science, University of Illinois Urbana-Champaign, Urbana, Illinois, 61801 U.S.A.\\
$^6$Departments of Mathematics and Economics, University of Illinois Urbana-Champaign, Urbana, Illinois, 61801 U.S.A.}
\end{center}
\medskip
\noindent{\bf Corresponding author:} Ruth Davidson, Departments of Mathematics and Plant Biology, University of Illinois Urbana-Champaign, Urbana, Illinois, 61801, U.S.A.; E-mail:  redavid2@illinois.edu \\

% Of course the specific format of addresses may vary according to
% country or other factors. Also, that was just an example email format.
%It's acceptable to add email addresses for authors in addition to the
%corresponding author. These would be placed after "Country."

%\vspace{1in}

%\begin{document}

%\title{A combinatorial method for connecting BHV spaces representing different numbers of taxa}

%\author{Jingwen~Bi, Ruth~Davidson, Michelle~Delcourt, Rosemary~Guzman, \\ Cara~Monical, Jos\'{e}~A.~Sanchez, Yingying~Ren, and Sihan~Zha}

%\maketitle

%\footnote{J. Bi, R. Davidson, D. Michelle, R. Guzman, C. Monical, J.A. Sanchez, Y. Ren, and S. Zha are with the University of Illinois Urbana-Champaign. email: bi4@illinois.edu, redavid2@illinois.edu, rguzma1@illinois.edu, cmonica2@illinois.edu, jasnchz2@illinois.edu, 	
%yren17@illinois.edu, 	
%shz@illinois.edu. M. Delcourt is with the University of Birmingham, UK. email:delcour2@illinois.edu. \textbf{Please direct all correspondence regarding this paper to redavid2@illinois.edu.}}

%\vspace{-1cm}

\begin{abstract}
The phylogenetic tree space introduced by Billera, Holmes, and Vogtmann ($\BHV$ tree space) is a CAT(0) continuous space that represents trees with edge weights with an intrinsic geodesic distance measure.  The geodesic distance measure unique to BHV tree space is well known to be computable in polynomial time, which makes it a potentially powerful tool for optimization problems in phylogenetics and phylogenomics.  Specifically, there is significant interest in comparing and combining phylogenetic trees.  For example, $\BHV$ tree space has been shown to be potentially useful in tree summary and consensus methods, which require combining trees with different number of leaves.  Yet an open problem is to transition between $\BHV$ tree spaces of different maximal dimension, where each maximal dimension corresponds to the complete set of edge-weighted trees with a fixed number of leaves.   We show a combinatorial method to transition between copies of $\BHV$ tree spaces in which trees with different numbers of taxa can be studied, derived from its topological structure and geometric properties.  This method removes obstacles for embedding problems such as supertree and consensus methods in the $\BHV$ treespace framework.
\end{abstract}

\begin{center}
supertree methods, consensus methods,  graph theory, CAT(0) spaces

\end{center}

%\section{Introduction}\label{Introduction\citep{}

Originally introduced in 200, the Billera, Holmes, and Vogtmann ($\BHV$) treespace  \citep{billera2001geometry} has long intrigued the mathematics and statistics communities, but performing computations relevant to the construction of the tree of life in this space that are of contemporary interest to the computational and systematic biology communities remains difficult for many reasons.  Yet significant progress has been made towards removing key obstacles to such computations, beginning with the software and polynomial-time algorithm introduced in \citep{owen2011fast}. This was a significant advance because $\BHV$ space is a CAT(0) space with an intrinsic geodesic distance measure, and the biological significance of how far apart two trees are is an essential issue for assessing the accuracy of phylogenies computed from both biological and simulated data. 

Traditional statistical analyses-which are key to assessing confidence levels in phylogeny estimation-are difficult to perform in $\BHV$ treespace as it is a non-Euclidean space \citep{benner2014point}. Yet there has been progress in the development of methods for performing statistical analyses in $\BHV$ treespace \citep{nye2011principal, barden2013central, miller2015polyhedral, weyenberg2016normalizing}. Further, continued exploration of the  geometric structure~\citep{lin2015convexity} and the use of such deeper understanding to improve optimization-based tree inference methods~\citep{skwerer2014optimization} are promising.  Much work remains to be done before the mathematical foundations of this space are fully explored to the extent where phylogeny reconstruction, evaluation, and related data analysis can be applied in this space. 

The contribution in this manuscript is to provide a combinatorial paradigm for relating copies of $\BHV$ treespace that correspond to trees with differing numbers of leaves as well as differing internal structures.  In particular, there is a unique copy of $\BHV$ treespace in which components of maximal dimension are determined by the number of internal edges of binary trees. This is an equivalent notion to identifying copies of $\BHV$ space that correspond to trees with $n$ leaves that are fully resolved; i.e. those not containing polytomies. We present a combinatorial method with a mathematical foundation for moving between copies of $\BHV$ space that can be identified with fully resolved phylogenies with $n$ leaves.  

In the first Section, ``Mathematical Foundations and Definitions", we give a broad overview of the mathematical foundation and common notation from previous publications, as well as novel definitions required for our results and new notation used in this paper. Our results underlying the combinatorial paradigm developed to transition between copies of $\BHV$ space designed to study sets of trees with different numbers of leaves are presented in Section entitled ``Results." In the ``Discussion" Section, we address the potential for applications of our combinatorial paradigm in computational biology that were not possible without a method to move between $\BHV$ spaces corresponding to trees with different numbers of taxa. 

%\ref{sec:math} \ref{sec:results} \ref{sec:discussion}RD  put this here in case we resubmit to a journal that allows Section numbering

\section{Mathematical Foundations and Definitions}\label{sec:math}
\subsection{Phylogenetic Trees}
A \emph{phylogeny} is a mathematical model of the common evolutionary of a group of taxa.  For example, the taxa may be genes, species, or individuals within a conspecific population. We adopt the convention for this manuscript that a phylogeny is a tree graph, and refer to phylogenies as \emph{phylogenetic trees}.  In a fully resolved phylogenetic tree the evolutionary history is represented by a tree $T$ with a label set assigned to the \emph{leaves} of degree one (which also represent the taxa under study) and each internal vertex of $T$ has degree of at least 3.  The internal vertex set represents the branching points in evolutionary history that result in taxon divergence due to evolutionary events such as point mutation, recombination \citep{kim2016structural}, or gene inversion \citep{francis2014algebraic}. Yet we adopt the perspective that despite these types of events, evolution is fundamentally treelike on a large scale, even when forces such as lateral gene transfer are the likeliest explanation for speciation history, which is supported by publications such as \citep{abby2012lateral}.  Further discussion of such important evolutionary events leading to non-treelike structures is not relevant to our results.  

This paper views all phylogenetic trees as non-rooted, and thus the tree structures show relative similarities and differences between species instead of an implied chronological order.  From this perspective, the topology (shape) of a phylogenetic tree with three leaves or fewer does not provide any biological information for the species described. Therefore, all definitions and theorems below focus on trees with four leaves or more. 

\begin{figure}[ht!]
\includegraphics[scale=.4]{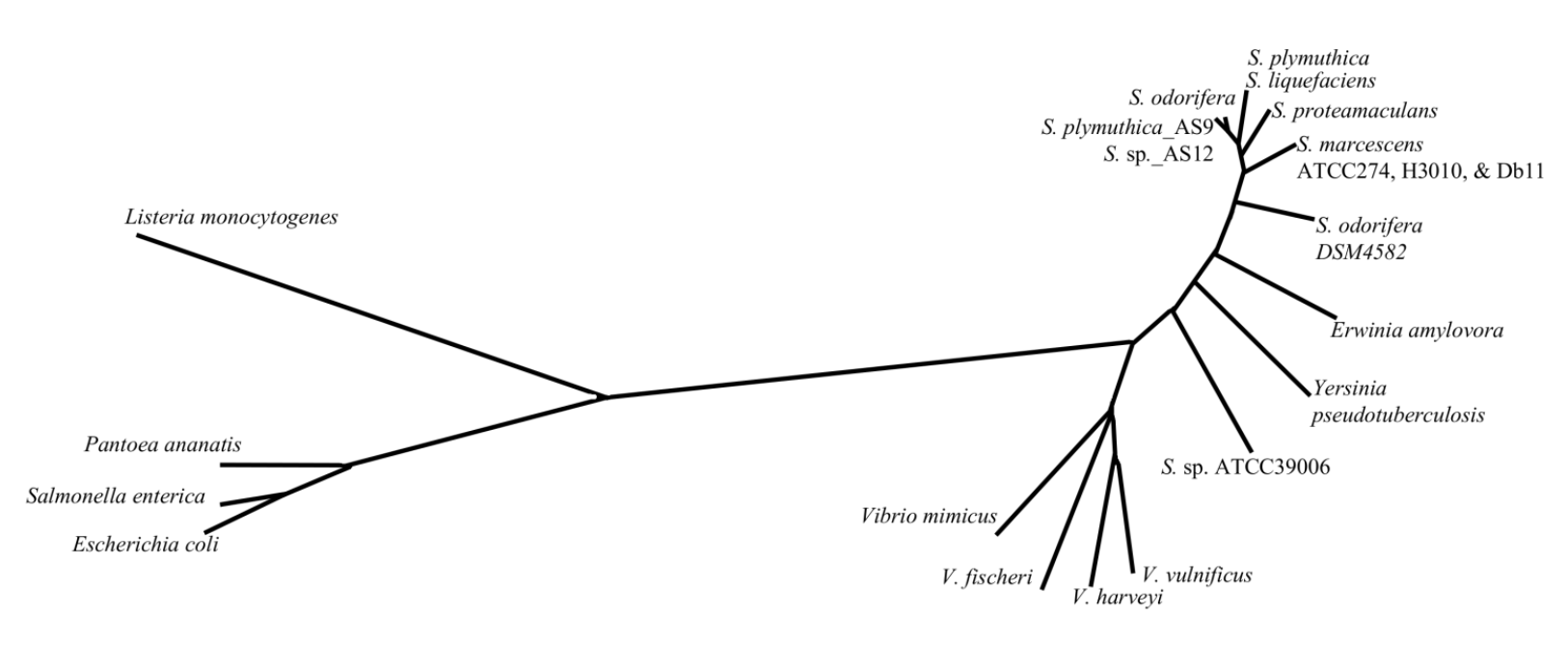}
\centering
\caption{Unrooted phylogenetic tree for the LuxS gene, from the Supplementary Material published in \citep{joyner2014use}}
\label{fig:phylogenetictree}
\centering
\end{figure}

For clarity, in this manuscript we use the same notation for trees as in \citep{owen2011fast}. A phylogenetic tree is a tree $T = (X, E, \Sigma)$, where $X$ is the label set assigned to the leaves of the tree,  $E$ is the set of interior edges, and $\Sigma$ is the set of splits of the set $X$ induced by the interior edges. In other words, the split $X_{e}|\overline X_{e}$ associated with edge $e \in E$ represents the partition of $X$ introduced by removing the edge $e$ from $T$.  

Following \citep{semple2003phylogenetics}, we say two splits associated with edges $\{e, f\} \in E$, $X_e|\overline X_e$ and $X_f|\overline X_f$, are \textbf{compatible} if one of the sets $$X_{e} \cap X_{f}, X_{e}\cap \overline X_{f}, \overline X_{e} \cap X_{f}, \overline X_{e}\cap \overline X_{f}$$ is empty. This is equivalent to asserting that one of the following set relationships is valid: $$X_{e} \subset X_{f}, X_{f} \subset X_{e}, X_{e} \subset \overline X_{f}, \overline X_{f} \subset X_{e}, \overline X_{e} \subset X_{f}, X_{f} \subset \overline X_{e}, \overline X_{e} \subset \overline X_{f},\overline X_{f} \subset \overline X_{e}.$$  Splits induced by a leaf edge are referred to as \emph{trivial}, as they provide no information from the perspective we adopt in this manuscript about the evolutionary relationships in the phylogeny. Figure \ref{fig:phylogenetictree} shows an unrooted phylogenetic tree inferred from biological data that was published in the supplementary materials for  \citep{joyner2014use}.  One can observe the biparititons of the taxa that are non-trivial induced by the internal edges of the tree. These correspond to set bipartitions of $X$ where each subset in $X_e|\overline X_e$ has cardinality at least two.

\subsection{Billera-Holmes-Vogtmann ($\BHV$) Tree Space}\label{sec:BHVspace}

$\BHV$ tree space is a continuous tree space that embeds trees using their split weights-where weights are the length of the internal edges corresponding to a non-trivial split. This tree space is formed by a set of Euclidean subspaces, called \emph{orthants} (a generalization of the notion of, for example, a quadrant in $\mathbb{R}^{2}$).  Each orthant of $\BHV$ space uniquely represents phylogenetic trees of different split weights but the same underlying topology. Orthants are joined together by lower dimensional orthants whenever the topologies they represent share common splits. Yet lower-dimensional orthants in a copy of $\BHV$ tree space correspond to tree topologies that have an internal vertex of degree higher than 3.  Therefore, paths between orthants of maximal dimension in $\BHV$ space that cross lower dimensional orthants can be thought of simply as collapsing internal edges for one labeled topology, thereby introducing a polytomy, and expanding the polytomy in a way that corresponds to a different labeled topology.

$\BHV$ space uses the \emph{geodesic} introduced in \citep{billera2001geometry} as its \emph{intrinsic} metric, which is defined as the shortest path between two points that lies completely inside the space. In  \citep{billera2001geometry} it was also shown that $\BHV$ tree space is CAT(0), or has globally non-positive curvature, and the geodesic is unique. The paper \citep{owen2011fast} introduced an algorithm with polynomial time complexity for computing the geodesic distance; this was a major advance towards making $\BHV$ space an object for the study of phylogenetic trees.

\begin{definition}\label{BHVn}
Denote a $\BHV$ tree space in which the maximum-dimensional orthants corresponding to fully resolved binary trees with $n$ taxa as $\BHV_{n}$. In other words, all internal vertices have degree three, such as in Figure \ref{fig:phylogenetictree}.
\end{definition}

In $\BHV_{n}$ the maximum-dimensional orthants have dimension $k$, where $k = n-3$ is the number of internal edges in a fully resolved binary tree with $n$ taxa. We briefly mention that $\BHV_{n}$ can be embedded in $\mathbb{R}^{N}$, where $N=2^{n-1}-n-1$ is the number of possible splits on $n$ taxa. This embedding is not useful without requiring the use of an \emph{extrinsic} metric, such as in \citep{lin2016tropical}.  Extrinsic metrics are useful in contexts beyond the scope of this paper.  Here we only mention this to clarify the difference between an embedding of $\BHV_{n}$ in another space for visualization purposes and mathematical foundations for different lines of research regarding $\BHV_{n}$. Our results only rely on the intrinsic metric of the geodesic distance in $\BHV_{n}$.  In other words, we follow the convention that outside of the CAT(0) surfaces that comprise $\BHV_{n}$, there is no mathematical information relevant to our results. 
\begin{figure}[ht!]
\includegraphics[scale=.2]{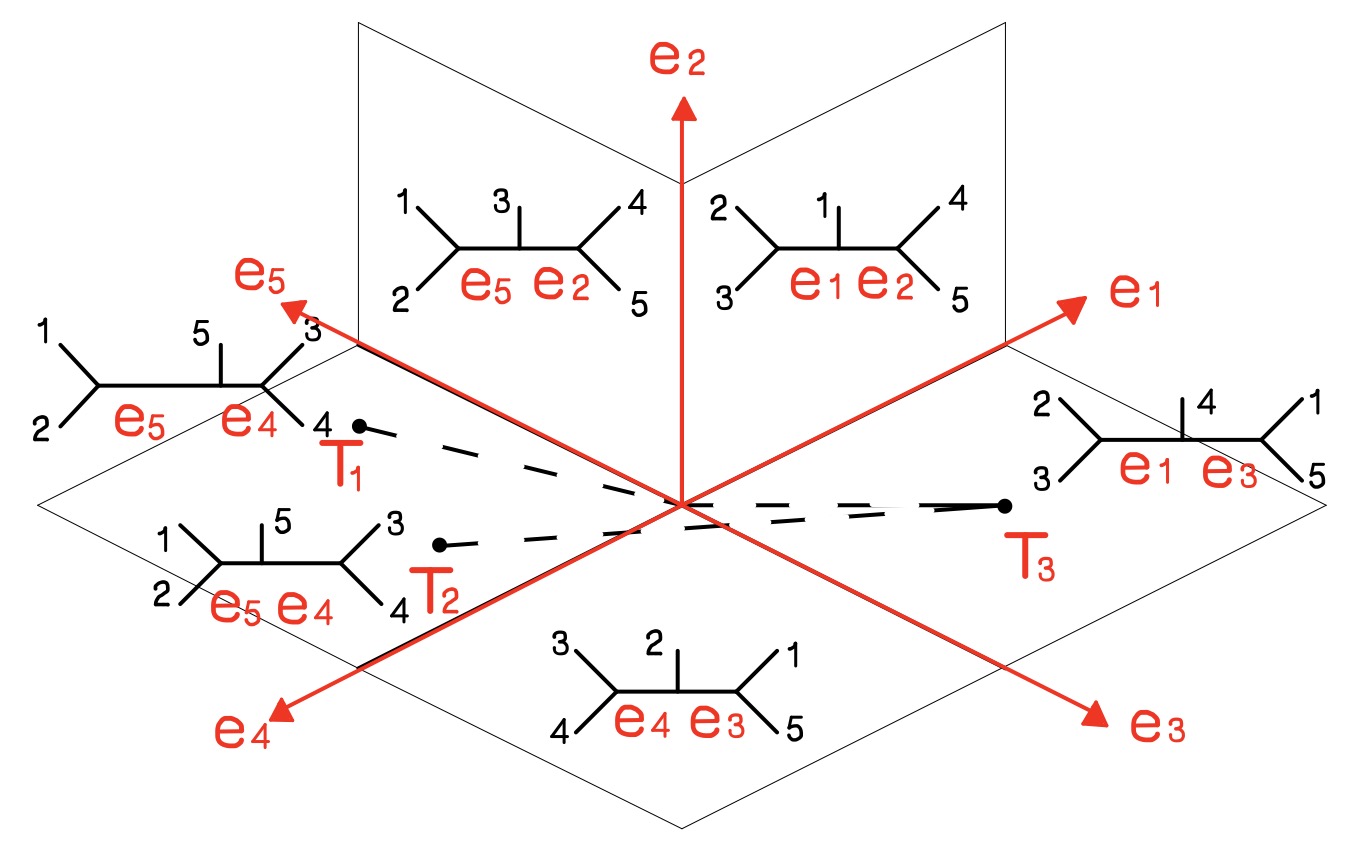}
\centering
\caption{Part of $\BHV_{5}$ embedded in $\mathbb{R}^3$ for visualization purposes:\\
There are fifteen two-dimensional orthants in $\BHV_{5}$ corresponding to the \\
fifteen labeled topologies on a five-leave unrooted phylogeny. \\ Note that the one-dimensional quadrants labeled with edges \\ correspond to the edges that collapse along the path traversed in the space. \\ All edges collapse at the origin into a star phylogeny.}
\label{fig:BHV5Megan}
\centering
\end{figure}

\subsection{Fundamental Definitions for our Results}

\begin{definition}\label{def:BHVconnectioncluster} A \textbf{$\BHV$ Connection Cluster} is defined upon the following collection of objects: (1) \emph{start tree} $T_{s}$, an unweighted binary phylogenetic tree, (2) \emph{start dimension} $n$, the number of leaves in the start tree, and (3) \emph{connection step} $\ell$, the number of new leaves added to the start tree. \end{definition} 
The $\BHV$ Connection Cluster denoted $C_{T_{s},n,\ell}$ is the set of unweighted fully resolved trees with $n+\ell$ leaves obtained from adding $\ell$ leaves to arbitrary edges (including leaf edges) of the start tree $T_{s}$ starting with $n$ leaves. Thus the new leaf set corresponding to this $\BHV$ Connection Cluster is the set containing the $n$ leaves of the start tree $T_{s}$ and the $\ell$ new leaves.  We note that the connection step $\ell$ creates a new tree containing both the splits present in the start tree $T_{s}$ in $\BHV_{n+\ell}$ and new splits induced by the new edges added by the connection step.  Hence, the maximum-dimensional orthants in $\BHV_{n + \ell}$ now have dimension $n+\ell -3$.

\begin{figure}[ht!]
\includegraphics[scale=.18]{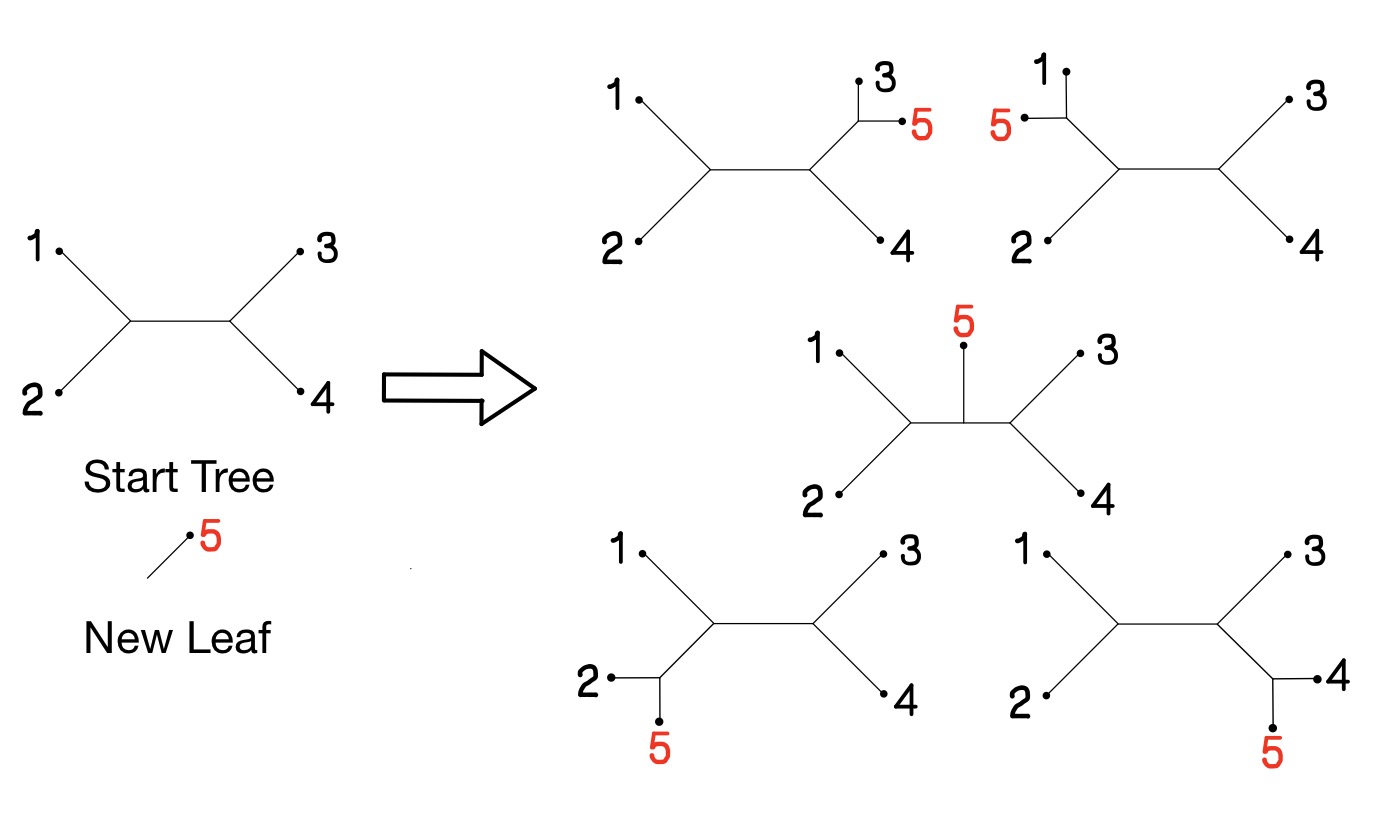}
\centering
\caption{The $\BHV$ Connection Cluster $C_{T_{s}, 4, 1}$ with start tree $T_{s} \in \BHV_{4}$}
\label{fig:BHVCluster}
\end{figure}

Trees in a $\BHV$ Connection Cluster, $C_{T_{s},n,\ell}$, are unweighted fully resolved $(n+\ell)$-leaf trees. If we assign arbitrary edge weights to every tree in the cluster, and define a split weights vector for each tree, we will obtain a set of $(n+\ell-3)$-dimension orthants in $\BHV_{n+\ell}$ which shares the same leaf set as the cluster.

\begin{definition}\label{def:BHVconnectionspace}
 Define this set of $(n+\ell-3)$-dimension orthants and the lower dimensional orthants associated with them as a \textbf{$\mathbf{\BHV}$ Connection Space}, or $S_{T_{s},n,\ell}$. Each point in this space represents a tree, and the coordinates of the point are the split weights of the tree. We define the \textbf{leaf set} of the $\BHV$ Connection Space to be the same leaf set as the corresponding Connection Cluster. \end{definition} Note that the one-dimensional orthants in this space will correspond to splits on the leaf set. We also refer to a one-dimensional orthant as the axis of the space. The distance between two points in the space is the same as the geodesic path for the two points in $\BHV_{n+\ell}$.  Intuitively, the $\BHV$ Connection Space is a group of Euclidean orthants glued together, and each individual orthant contains trees with different edge weights but same topology, as in $\BHV$ tree space.
\begin{figure}[ht!]
\includegraphics[scale=.18]{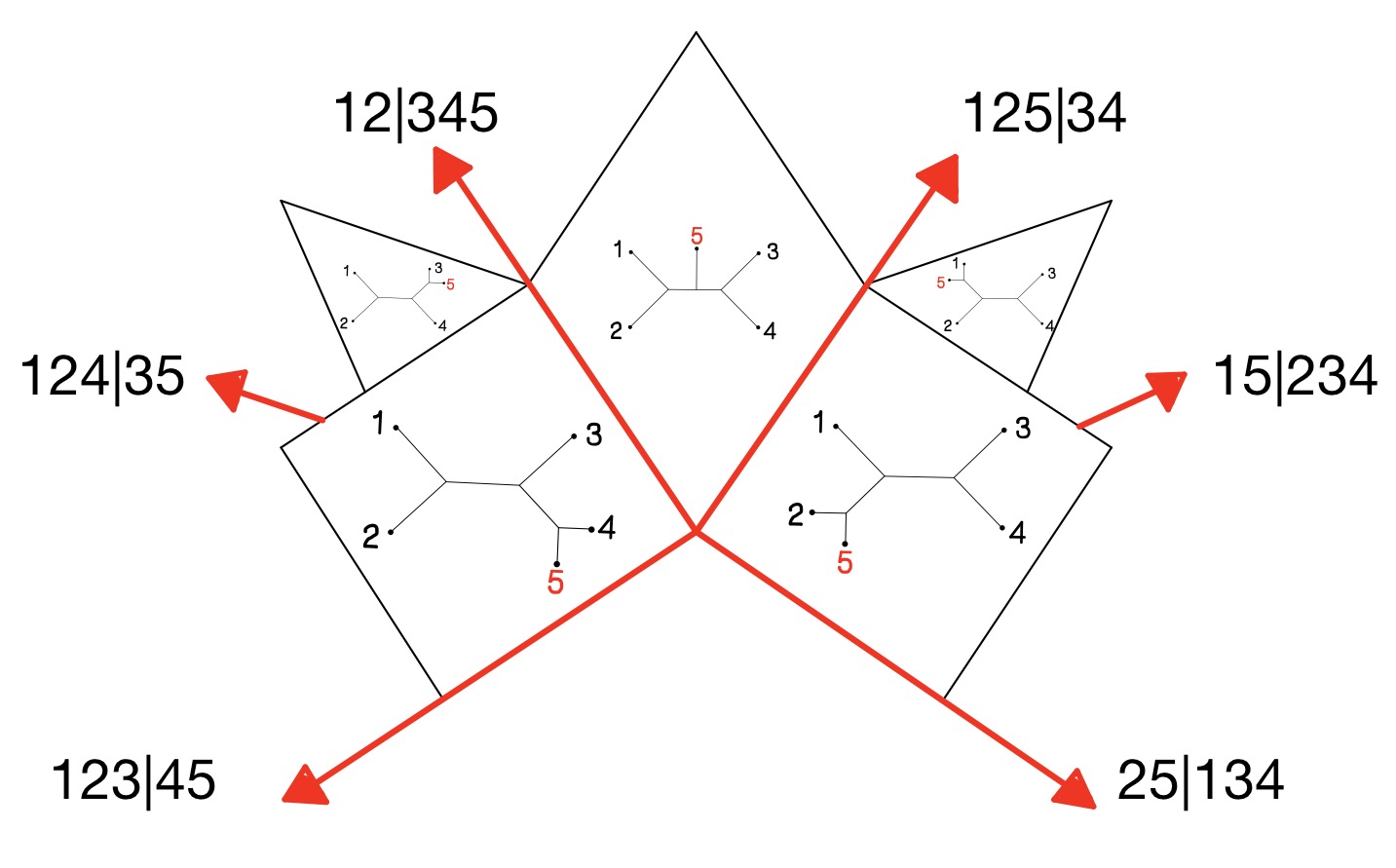}
\centering
\caption{The $\BHV$ Connection Space $S_{T_{4}, 4, 1}$}
\label{fig:BHVSpace}
\centering
\end{figure}
\begin{definition}\label{def:BHVconnectiongraph} A \textbf{$\mathbf{\BHV}$ connection graph}, or $G_{T_{s},n,\ell}$ has vertex set $V_{C}$: the set of one-dimensional orthants in $S_{T_{s},n,\ell}$. As mentioned earlier, $V_{C}$ is equivalently the set of splits in all the trees in $C_{T_{s}, n, \ell}$. The edge set of $G_{T_{s}, n, \ell}$ is $E$: \{$v_{i}v_{j}\}$: $v_{i},v_{j}\in V_{C}$ and there exists a tree $T \in C_{T_{s},n,\ell}$ such that $T$ contains the two sets of splits represented by $\{v_{i}$ and $v_{j}\}$. The leaf set of $G_{T_{s},n,\ell}$ is the same as the leaf set of $S_{T_{s},n,\ell}$.
\end{definition}
\begin{figure}[ht!]
\includegraphics[scale=.10]{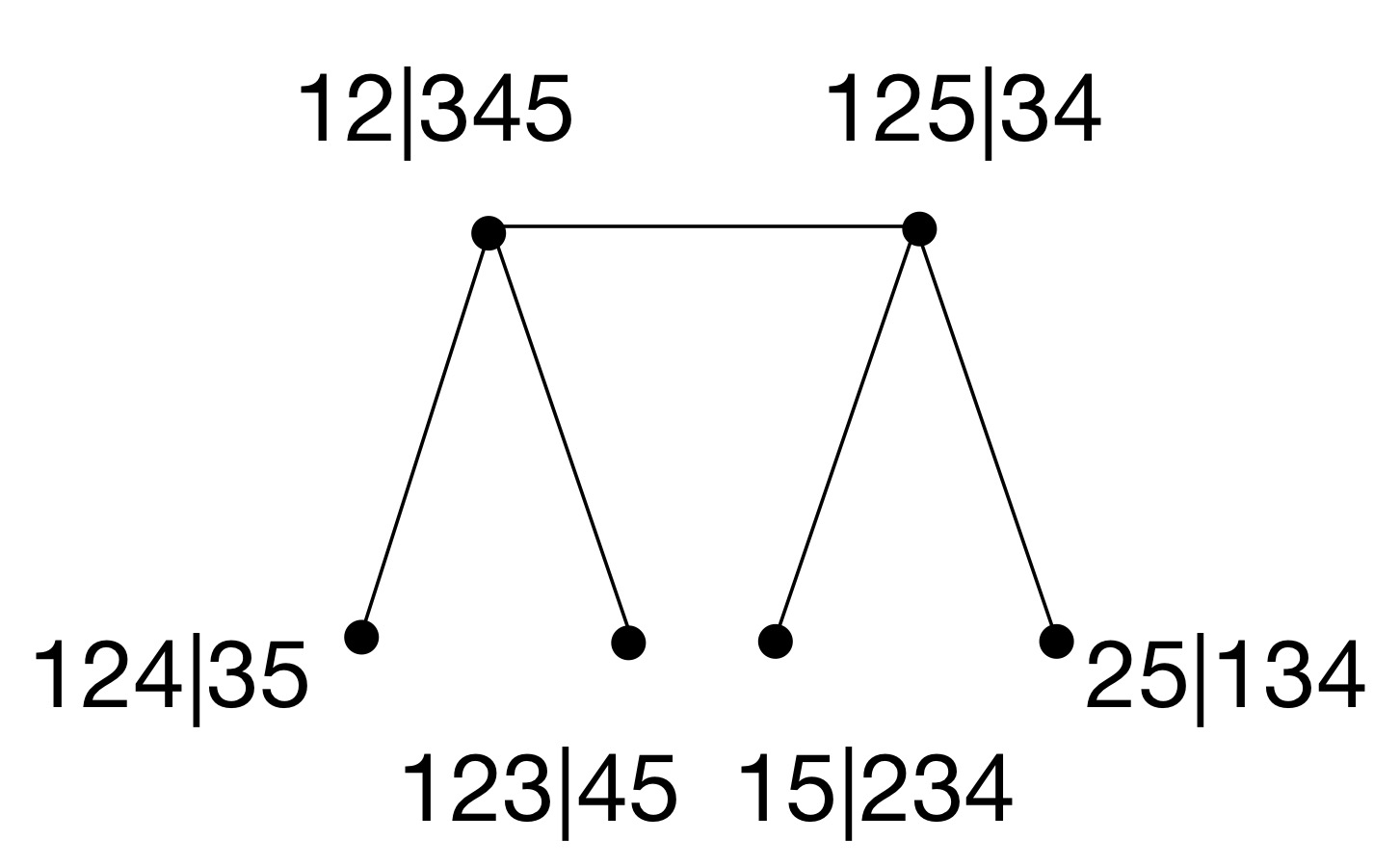}
\centering
\caption{The $\BHV$ Connection Graph $G_{T_{s}, 4, 1}$, where $T_{s} \in \BHV_{4}$}
\label{fig:BHVConnnectionGraph}
\centering
\end{figure}

\section{Results}\label{sec:results}

\begin{thm}\label{thm:Cardinalityofthecluster} The number of trees in a $\BHV$ Connection Cluster $C_{T_{s},n,\ell}$ is $\frac{(2(n+\ell)-5)!!}{(2n-5)!!}.$
\end{thm}
\begin{proof}
A fully-resolved $m$-leaf tree $T$ has $2m-3$ edges, consisting of $m$ leaf edges and $m-3$ internal edges; thus there are $2m-3$ different fully-resolved $(m+1)$-leaf trees that can be obtained from adding a single new leaf to the $m$-leaf start tree. Adding new leaves to distinct original edges are probabilistically independent events. Thus the number of $(n+\ell)$-leaf trees obtained by adding $\ell$ new leaves to an $n$-leaf start tree $T_{s}$ is
\begin{eqnarray*}
&(2n-3)& \times (2n-1)\times \ldots \times(2(n+\ell-1)-3)\\
&=&\frac{1\times3\times \ldots \times(2n-5)\times(2n-3)\times(2n-1)\times \ldots \times(2(n+\ell-1)-3)}{1\times3\times \ldots \times(2n-5)}\\
&=&\frac{(2(n+\ell-1)-3)!!}{(2n-5)!!} = \frac{(2(n+\ell)-5)!!}{(2n-5)!!}.
\end{eqnarray*}
\end{proof}

\begin{thm}\label{thm:dimensionofthespace} $\BHV$ Connection Space $S_{T_{s},n,\ell}$ has dimension $2^{\ell}(2n-2)-\ell-n-1$.
\end{thm}

\begin{proof}
$\BHV$ Connection Space $S_{T_{s},n,\ell}$ has dimension equal to the number of unique nontrivial splits from all the trees in the corresponding $\BHV$ Connection Cluster, $C_{T_{s},n,\ell}$.  Because each split has the form $X_e|\overline X_e$, a leaf will either belong to $X_e$ or $\overline X_e$. If we assign label $1$ or $2$ to every leaf, then the set of leaves with label $1$, $X_1$, and the set of leaves with label $2$, $X_2$, will form a split $X_1|X_2$ on the leaf set. 
% (Note we are not using the overline notation $X_{e}|\overline{X_{e}$ here-as we are assigning duplicate labels to $X$). 
The new tree in $\BHV_{n +\ell}$ is denoted $T_{t}.$

We have $2^{n+\ell}$ different ways of assigning the labels $1$ or $2$ the new leaf set of size $n+\ell$, but there are certain constraints for a label assignment split to be a tree split. Consider the following three cases:
\begin{enumerate}
    \item We first consider splits where all leaves from the start tree are in one side of the split. Namely, if $V$ is the leaf set of $T_{s}$ and $L$ is the new leaf set, then in this first case, we are counting all splits with the form $$Y_e|\overline Y_e, V\subset Y_e, Y_e\cup \overline Y_e = L.$$ 
    Without loss of generality, assign label $1$ to all leaves from the original start tree. There are $2^\ell$ choices to assign a label to the $\ell$ new leaves in $T_{t}.$ However, the choice to assign label $1$ to all $\ell$ new leaves in $T_{t}$ results in all leaves in the tree having the same label, which will not define a split. Moreover, if $\ell-1$ new leaves get label $1$ in $T_{t}$, and the remaining one leaf gets label $2$, then we have a split defined by the leaf edge, which is a trivial split. Since we have $\ell$ new leaves, there are $\ell$ assignments resulting in this situation. So we are left with $2^{\ell}-1-\ell$ different label assignments for $T_{t}$.
    \item We consider the splits that are derived from the trivial leaf splits of the start tree $T_{s}$ by adding the new $\ell$ leaves to $T_{s}$ to obtain $T_{t}$. Again, let $V$ be the leaf set of $T_{s}$ and $L$ be the new leaf set ($|L| = \ell$) in $T_{t}.$  We are counting splits of the form $$(\{x \in X \}\cup Y_{e})|(Y_{e} \cup V\setminus \{v\}), v\in V, Y_e\cup \overline Y_{e} = L.$$
    Assign label $1$ to one leaf from the start tree and label $2$ to all remaining leaves in the start tree. Denote the leaf with label $1$ as $v_{1}$. We again have $2^{\ell}$ choices for assigning labels to the $\ell$ new leaves in $T_{t}$. However, if all $\ell$ new leaves get label $2$, we have a trivial leaf edge split between $v_{1}$ and the rest of the leaves. So we have $2^{\ell}-1$ choices. Since we choose $v_{1}$ randomly, we have $n$ different choices for $v_{1}$. Therefore, we have $(2^{\ell}-1)n$ different label assignments.
    \item In the last case, we consider splits that are derived from a nontrivial leaf split of $T_{s}$ by adding the new $\ell$ leaves to obtain $T_{t}$. Using the same notation as in the previous case, we are counting splits in the form of $$(Y_e\cup Z_e)|(\overline Y_e \cup \overline Z_e), Y_e\cup\overline Y_e = L, Z_e\cup\overline Z_e = V, |Z_e|\geq2, |\overline Z_e|\geq 2.$$
    Recall that $T_{s}$ is a fully resolved binary tree with $n$ leaves, thus the start tree defines $n-3$ unique nontrivial splits. For each split, label the leaves from different sides of the split with label $1$ and $2$ respectively. Then we have $2^{\ell}$ different ways of assigning labels for the $\ell$ new leaves. So there are $(n-3)2^{\ell}$ different label assignments for $T_{t}$.
\end{enumerate}

To see these three cases are mutually exclusive to each other, note that we have different label assignments for the leaves in $T_{s}$ in each case. Specifically, we let the leaves in $T_{s}$ share in the same label in case (1), assign a different label to a single leaf from the rest of the leaves in case (2), and assign different labels based on a split of $T_{s}$ in case (3). These three cases also cover all possible assignments for $T_{s}$. Therefore, the total number of splits in $T_{t}$ is \[2^{\ell}-1-\ell+(2^{\ell}-1)n+(n-3)2^{\ell} = 
2^{\ell}(2n-2)-\ell-n-1.\]
\end{proof}

Because a three-leaf unweighted tree in $\BHV_{3}$ does not provide any biological information, the $\BHV$ Connection Space with start tree $T_{s} \in \BHV_{3}$ and connection step $\ell$ is the same space as $\BHV_{l+3}$ with the same leaf set as the Connection Space. To observe this more carefully, we can compute the number of orthants and dimension of $S_{T_{s},3,\ell}, T_{s} \in \BHV_{3}$ using the formulas from Theorems \ref{thm:Cardinalityofthecluster} and \ref{thm:dimensionofthespace}: the number of maximum-dimensional orthants is $$\frac{(2(3+\ell)-5)!!}{(2\cdot 3-5)!!} = \frac{(2(3+\ell)-5)!!}{1!!} = (2(\ell+3)-5)!!,$$ and the dimension of these orthants is $$2^{\ell}(2\cdot 3-2)-\ell-3-1 = 2^{(\ell+3)-1}-(\ell+3)-1.$$  We remark that our formulas do not conflict with the dimension calculations of $\BHV_{\ell+3}$ given in \citep{billera2001geometry} and \citep{owen2011fast}.
\\

\begin{lem}\label{lem:compatiblesplitsingraph}For any two vertices in $\BHV$ Connection Graph $G_{T_{s},n,\ell}$, if they represent two compatible splits, then they are connected in the $\BHV$ Connection Graph.
\end{lem}
\begin{proof}
For any two vertices, say $v_{1}, v_{2}$, in the $\BHV$ Connection Graph $G_{T_{s},n,\ell}$ (we will refer to this graph as $G_{C}$ for the remainder of this proof) that represent two compatible splits, to prove that they are connected in $G_{C}$ is equivalent to proving that there exists a tree in the corresponding $\BHV$ Connection Cluster that contains both of the splits that $v_{1}$ and $v_2$ represent. Denote the new leaf set in $G_{C}$ as $V_{new}, |V_{new}| = \ell$. Removing a common leaf from two compatible splits will result in two new compatible splits. 

We initialize a list for pairs of splits. $v_{1}$ and $v_{2}$ are splits on $n+\ell$ leaves. Denote $P_{n+\ell} = (v_{1}, v_{2})$. Let $P_{n+\ell}$ be the beginning of the list. The vertices $v_{1}$ and $v_{2}$ both contain all leaves in $V_{new}$. We generate the two compatible splits in $P_{n+\ell-1}$ by removing a leaf in $V_{new}$ from $v_{1}$ and $v_{2}$. In fact, we can generate the two compatible splits in $P_{n+i}$ from $P_{n+i+1}$ for all $i \geq 0$ by removing one leaf that is in $V_{new}$ and also in the two splits in $P_{n+i+1}$ from the two splits in $P_{n+i+1}$. So we will have a list $\{P_{n+\ell}, P_{n+\ell-1}, \ldots, P_{n+1}, P_{n}\}$, where $P_{n+\ell}$ is the beginning of the list and $P_{n}$ is the end of the list. Notice that the two compatible splits in $P_{n}$ are splits on the leaf set of $T_{s}$, and the splits in $P_{n+i+1}$ will contain one more leaf than the splits in $P_{n+i}$. That leaf is in $V_{new}$.   To prove that there exists a tree that is formed by adding the $\ell$ new leaves to the start tree $T_{s}$ \emph{and} contains both of the splits in $P_{n+\ell}$, we proceed by induction on $P_{n+i}$.  When $i = 0$, $T_{s}$ contains both of the splits in $P_{n}$. 

As in our inductive hypothesis, assume there is a tree $T_{i}$ that contains both of the splits in $P_{n+i}$ and the extra leaf in the splits in $P_{n+i+1}$ is denoted as $r$. If the two splits in $P_{n+i}$ are $X_{e_{1}}|\overline X_{e_{1}}$ and $X_{e_{2}}|\overline X_{e_{2}}$, then without loss of generality, the two splits in $P_{n+i+1}$ will be $X_{e_{1}}\cup \{r\}|\overline X_{e_{1}}$ and $X_{e_{2}}\cup\{r\}|\overline X_{e_{2}}$. By the  construction of all elements in the list, we know the two splits in $P_{n+i+1}$ are compatible. So one and only one of $$(X_{e_{1}}\cup \{r\})\cap (X_{e_{2}}\cup\{r\}), (X_{e_{1}}\cup \{r\})\cap (\overline X_{e_{2}}),  (\overline X_{e_{1}})\cap (X_{e_{2}}\cup\{r\}), (\overline X_{e_{1}})\cap (\overline X_{e_{2}})$$ has to be an empty set. Since  $r\in (X_{e_{1}}\cup \{r\})\cap (X_{e_{2}}\cup\{r\})$, one of $$(X_{e_{1}}\cup \{r\})\cap (\overline X_{e_{2}}) = (X_{e_{1}})\cap (\overline X_{e_{2}}),  (\overline X_{e_{1}})\cap (X_{e_{2}}\cup\{r\})=(\overline X_{e_{1}})\cap (X_{e_{2}}), (\overline X_{e_{1}})\cap (\overline X_{e_{2}})$$ has to be an empty set. On the other hand, the two splits in $P_{n+i}$ are also compatible. 

So one and only one of $$(X_{e_{1}})\cap (X_{e_{2}}), (X_{e_{1}})\cap (\overline X_{e_{2}}), (\overline X_{e_{1}})\cap (X_{e_{2}}), (\overline X_{e_{1}})\cap (\overline X_{e_{2}})$$ is non-empty. Thus $(X_{e_{1}})\cap (X_{e_{2}})$ is not an empty set. Denote one of the common leaves in $(X_{e_{1}})$ and $(X_{e_{2}})$ as $u$, where $u\in T_{i}$. If we choose any internal edge $e$ in $T_{i}$ connected to the parent of $u$,  we can add a leaf edge to $e$ with the leaf vertex as $r$, and the new tree (denoted $T_{i+1}$) that we obtain will contain both splits in $P_{n+i+1}$. Therefore, there exists a tree in the $\BHV$ Connection Cluster that contains both of the splits represented by $v_{1}$ and $v_{2}$.
\end{proof}

\begin{cor}\label{cor:compatibleorthantingraph}
Let $G$ be a $\BHV$ Connection Graph.  For any $V \subseteq V(G)$, if $V$ represents a set of compatible splits, then $V$ induces a complete subgraph of $G$, and thus the set of compatible splits form an orthant of dimension $|V|$ in the corresponding $\BHV$ Connection Space $S_{T_{s},n,\ell}$.
\end{cor}

\begin{proof}
We proceed by induction.  If $|V| = 1$, then the only vertex in the set is a complete graph in $G$, and it corresponds to an axis in $S_{T_{s},n,\ell}$ by Definition \ref{def:BHVconnectionspace}.  As in our inductive hypothesis, we assume that if $V \subseteq V(G)$ of size $n$ represents a set of compatible splits, then $V$ induces a complete subgraph of $G$, and the set of compatible splits form an orthant of dimension $n$ in $S_{T_{s},n,\ell}$.  For any $U \subseteq V(G)$ of size $n+1$ that represents a set of compatible splits, choose any $n$ vertices from $U$; they will induce a complete subgraph of $G$. The one remaining vertex, $v$, represents a split that is compatible with splits represented by $U \setminus \{v\}$. By Lemma~\ref{lem:compatiblesplitsingraph}, $v$ and any vertex from $U \setminus \{v\}$ is connected to all vertices in the subgraph with vertices in the set $U \setminus \{v\}$. Thus, $U$ induces a complete subgraph in $G$. By Definition \ref{def:BHVconnectiongraph}, the orthants formed by the splits represented by $U$ exist in $S_{T_{s},n,\ell}$.

\end{proof}
\begin{lem}\label{lem:splitsofgeodesic}For all $ n\geq 3$, the trees along the geodesic between any two trees $T_{1}, T_{2}$ in $BHV_{n}$ will only contain some split if that split is contained in either of the two trees.
\end{lem}
\begin{proof}
By \citep{billera2001geometry} Proposition 4.1, the geodesic between $T_{1}$ and $T_{2}$ traverses a sequence of orthants whose split set is a subset of the union of the split set of $T_{1}$ and that of $T_{2}$. 
\end{proof}

\begin{thm}\label{thm:subspace} The $\BHV$ Connection Space $S_{T_{s},n,\ell}$ is a convex space. In particular, the geodesic between any two trees $T$ and $T'$ in $S_{T_{s},n,\ell}$ lies within $S_{T_{s},n,\ell}$.
\end{thm}

\begin{proof}
Denote the split sets of $T$ and $T'$ as $\Sigma$ and $\Sigma '$. By Lemma~\ref{lem:splitsofgeodesic}, for any tree $T_{g}$ on the geodesic, $T_{g}$ only contain splits in $\Sigma \cup \Sigma '$. Thus the split set of $T_{g}$, $\Sigma_{g}$ is a subset of the vertex of $\BHV$ Connection Graph $G_{T_{s}, n, \ell}$. On the other hand, since splits in $\Sigma_{g}$ exist in a tree, they are compatible with each other. By Lemma~ \ref{lem:compatiblesplitsingraph}, the orthants formed by $\Sigma_{g}$ exist in $S_{T_{s},n,\ell}$ and thus $T_{g}$ exists in $S_{T_{s},n,\ell}$. Since $T_{g}$ is chosen arbitrarily along the geodesic, the geodesic is contained in $S_{T_{s},n,\ell}$ (as shown in \citep{owen2011fast} regarding the properties of the geodesic distance in any $\BHV$ space).
\end{proof}
\begin{figure}[ht!]
\includegraphics[scale=.2]{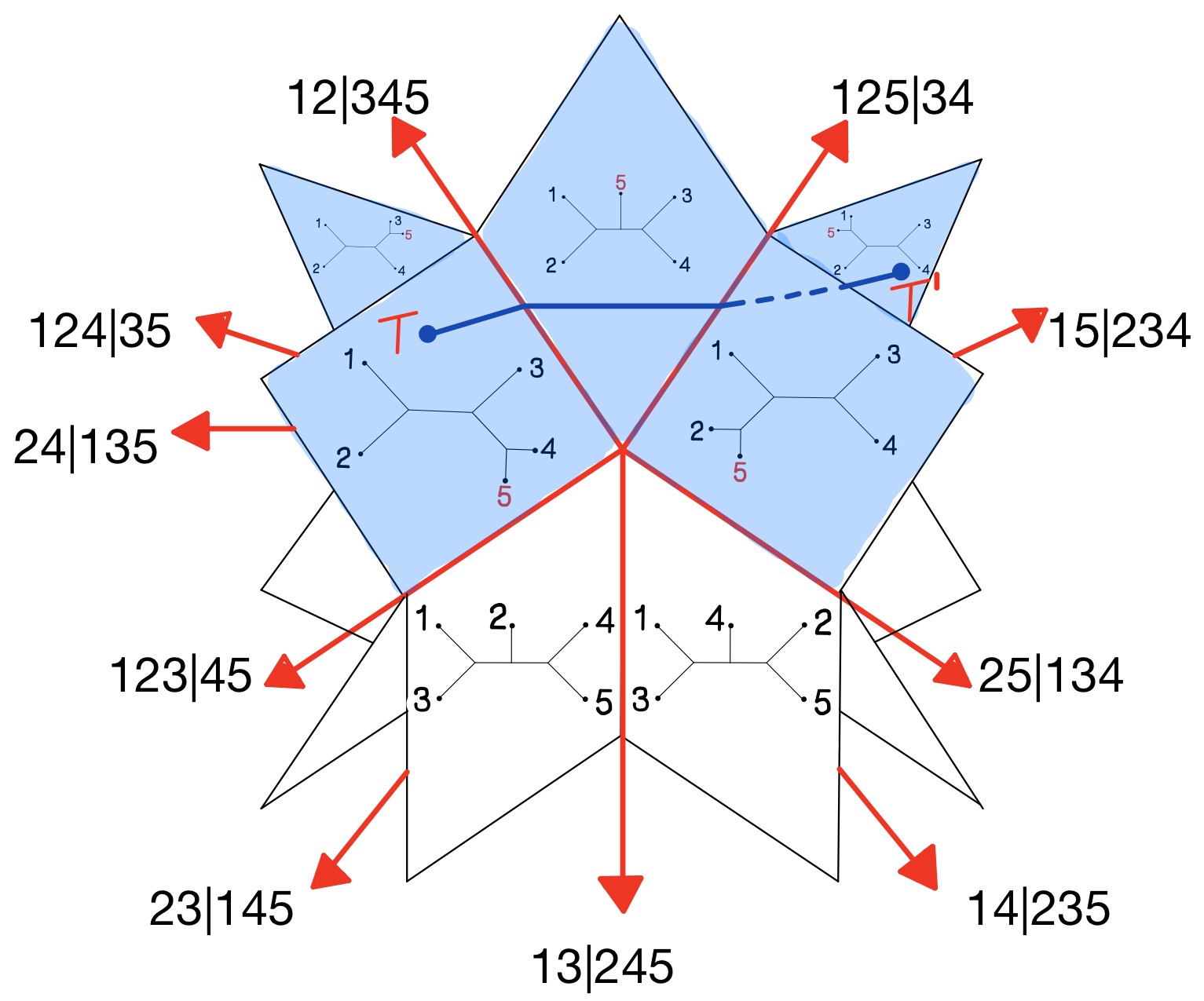}
\centering
\caption{An example of the geodesic path in $\BHV$ Connection Space (colored blue). The $\BHV$ Connection Space $S_{T_{s}, 4,1}$ is shown as part of $\BHV_{5}$}
\label{fig:BHVGraph}
\centering
\end{figure}

\subsection{$\BHV$ Connection Graphs with Connection Step 1}

For the rest of this section, we focus on $\BHV$ Connection Graphs with connection step 1.  We show that we can build $\BHV$ Connection Graphs for trees with complicated shapes from simpler trees.  

\begin{definition}
A \textit{caterpillar} is a an unrooted tree that can be represented in the plane by a graph where all the leaves have exactly one edge to a single line.
\end{definition}

\begin{figure}[ht!]
\includegraphics[scale=.10]{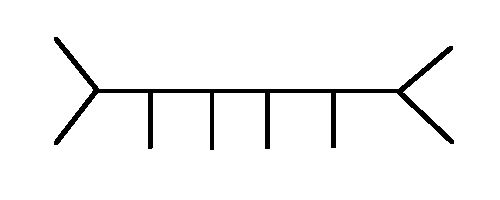}
\centering
\caption{A caterpillar with 8 leaves.}
\label{fig:ruthcaterpillar}
\centering
\end{figure}

Equivalently, caterpillars are trees where every vertex of degree at least three has at most two non-leaf neighbors. Figure \ref{fig:ruthcaterpillar} shows a caterpillar with 8 leaves. Binary caterpillars are caterpillars in which all non-leaf vertices have degree three.  Due to the simplicity of this tree shape, we are able to easily construct $\BHV$ Connection Graphs with connection step $1$ for binary caterpillars.  

We also observe that such $\BHV$ Connection Graphs for all binary trees can be constructed using binary caterpillar trees.  The concatenation method proposed in Theorem~\ref{thm:subgraph} provides a general idea of the structure of the $\BHV$ Connection Graph for trees with more complex topologies and arbitrary numbers of leaves.

\begin{thm}\label{thm:independentsetinGwhenk=1} The size of a largest independent set in the $\BHV$ Connection Graph with step 1, $G_{T_{s},n,1}$, is $n$.  Furthermore, the set of vertices representing splits introduced by adding the new leaf to a leaf edge is the only independent set of $G_{T_{s},n,1}$ of size $n$.
\end{thm}
\begin{proof}
When adding a single leaf to a tree, we may either add it to a leaf edge or an internal edge. Suppose the leaf set of $T_{s}$ is $W$, and $|W|=n$. Adding a new leaf $v$ to an existing leaf edge in $T_{s}$ will introduce splits such as $\{u, v\}|\{W\setminus \{u, v\}\}$ where $u\in W$. Adding a new leaf $v$ to an internal edge will introduce splits such as $\{v, A\}|\{W\setminus A\}$ where $A\subset W, |A| >1, |W\setminus A| > 1$. 

By definition, each vertex in a $\BHV$ Connection Graph corresponds to a split introduced by adding new leaves. For the remainder of the proof we will refer to the vertex set of $G_{T_{s},n,1}$ as $V$. We partition $V$ into parts:  the vertices which represent splits introduced by adding the new leaf to a leaf edge, denoted $V_{\ell}$, and vertices which represent splits introduced by adding the new leaf to an internal edge, denoted $V_{i}$.  Here $V_{i} \cup V_{\ell} = V, V_{i} \cap V_{\ell} = \varnothing.$ Moreover, adding a new leaf to an internal edge will introduce two splits so $|V_{i}| = 2(n-3)=2n-6$ and $|V_{\ell}| = n.$

Let $v$ be the new leaf added to the start tree. For all $\{u, m\} \subset W$, we know $\{u, v\}|\{W\setminus \{u, v\}\}$ and $\{m, v\}|\{W\setminus \{m, v\}\}$ are not compatible. Therefore, $V_{\ell}$ is an independent set of size $n$ in $G_{T_{s},n,1}$. The following proves that this is also the unique largest independent set in $G_{T_{s},n,1}$.

Assume for contradiction that the largest independent set in $G_{T_{s},n,1}$ contains vertices from $V_{i}$ and has size larger than $n$.
Define the largest independent set in $G_{T_{s},n,1}$ as $V_{\alpha}$. $|V_\alpha| \geq n$.
Any two vertices in $V_{\alpha}$ are not connected.  By Lemma~\ref{lem:compatiblesplitsingraph}, they correspond to a pair of incompatible splits. Since the pair of splits introduced by adding a new leaf to an internal edge are compatible, only one of the splits will have a corresponding vertex inside $V_{\alpha}$. Thus, $V_{\alpha}$ will have at most $n-3$ vertices from $V_{i}$. Denote some independent set in the $V_{i}$ induced subgraph of $G_{T_{s},n,1}$ as $V_{\beta}$.  Note $|V_{\beta}| \leq n-3$.

Each vertex in $V_{\beta}$ represents a split in the form $v \cup A_{i} | W \setminus A_{i}$, where $A_{i}\in V, |A_{i}|>1$, and $|W\setminus A_{i}|>1$. The splits introduced by adding the new leaf to an existing leaf edge which are compatible with $(\{v\} \cup A_{i}) | (W \setminus A_i)$ have the form $\{v, u \}|(W \setminus  \{v, u\})$, where  $u\in A_{i}$. Thus, the number of splits introduced by adding the new leaf to a leaf edge that are also compatible with $(v \cup A_i) |(W \setminus A_i)$ is $|A_i|$. 

\begin{figure}[ht!]
 
\begin{subfigure}{0.52\textwidth}
\includegraphics[scale=.17]{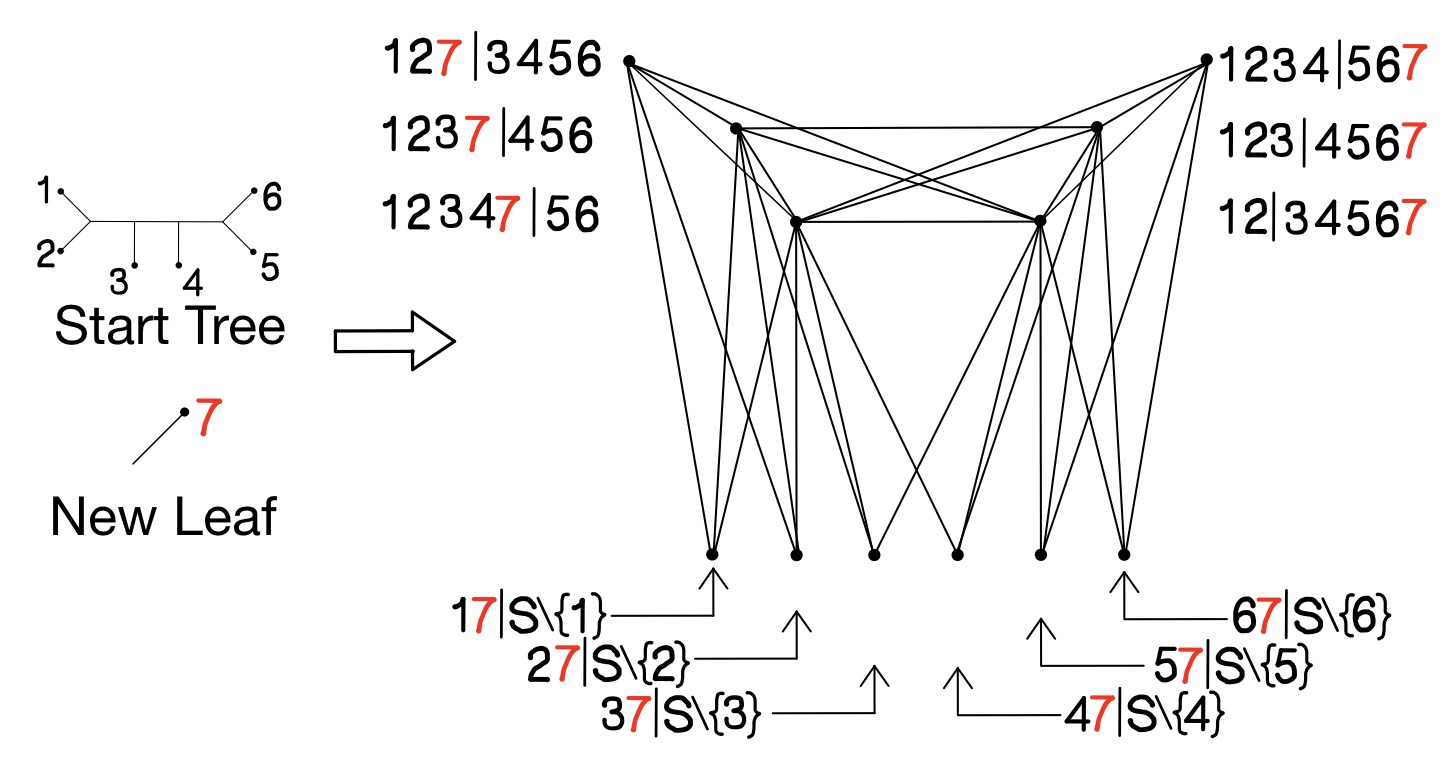}
\centering
\caption{}
\label{fig:BHVGCat6}
\end{subfigure}
\begin{subfigure}{0.5\textwidth}
\includegraphics[scale=.17]{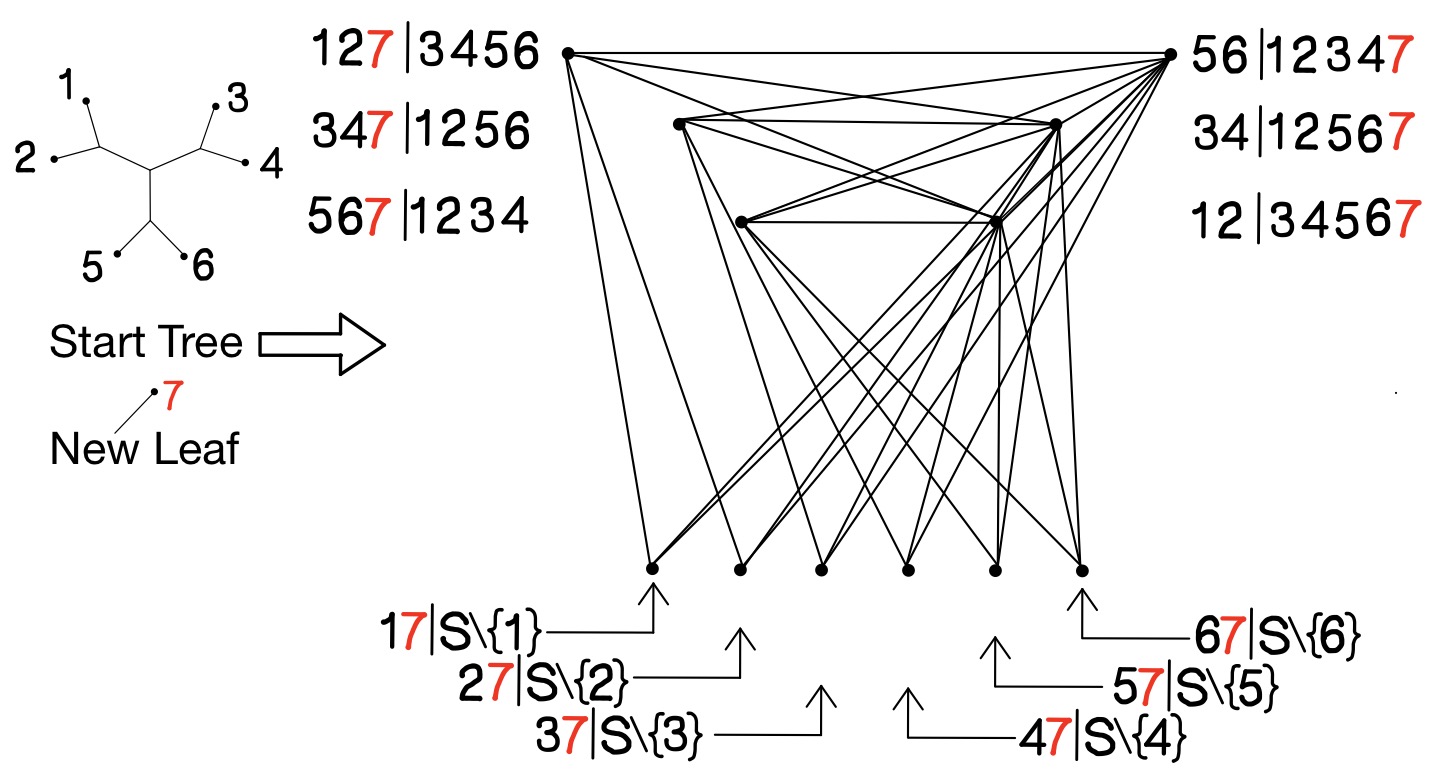}
\centering
\caption{}
\label{fig:BHVGFea6}
\end{subfigure}
\caption{The $\BHV$ Connection Graphs $G_{T_{(a)}, 6, 1}$ and $G_{T_{(b)}, 6, 1}$. In Figures (a) and (b), the vertex set of the start tree $T_{s}$ is $\{1, 2, 3, 4, 5, 6\}$. The bottom six vertices are $V_{\ell}$ in this graph, which also form an independent set. The start trees $T_{(a)}$ and $T_{(b)}$ have the same leaf set but different tree topologies. $G_{T_{(a)}, 6, 1}$ and $G_{T_{(b)}, 6, 1}$ are not isomorphic to each other as $G_{T_{(a)}, 6, 1}$ has vertex degree sequence $\{3, 3, 3, 3, 3, 3, 5, 5, 7, 7, 9, 9 \}$. $G_{T_{(b)}, 6, 1}$ has vertex degree sequence $\{3, 3, 3, 3, 3, 3, 5, 5, 5, 9, 9, 9\}$.}

\label{fig:BHVCG6}
\end{figure}

To calculate the total number of vertices in $V_{\ell}$ to which vertices in $V_{\beta}$ are adjacent, we only need to calculate the total number of splits introduced by adding the new leaf to a leaf edge that are also compatible with splits represented by vertices in $V_{\beta}$: $|\bigcup_{i=1}^{|V_{\beta}|} A_{i}|$. By the Inclusion - Exclusion principle, \[\mid\bigcup_{i=1}^{|V_{\beta}|} A_{i}| \geq \sum _{i=1}^{|S_{\beta}|} |A_{i}| - \sum _{1\leq i < j \leq |S_{\beta}|} |A_{i}\cap A_{j}|.\]

If $v \cup A_{i} | W \setminus A_{i}$ are leaves with trivial splits incompatible with those corresponding to $v \cup A_{j} | W \setminus A_{j}$ (which is the case for any two splits in $V_{\beta}$) then as in \citep{owen2011fast}: $$(v \cup A_{i}) \cap (W \setminus A_{j}) \neq \varnothing,$$
$$(W \setminus A_{i}) \cap (v \cup A_{j}) \neq \varnothing,$$
$$(W \setminus A_{i}) \cap (W \setminus A_{j}) \neq \varnothing.$$
Since $v \not \in W,$ $$(A_{i}) \cap (W \setminus A_{j}) \neq \varnothing,$$
$$(W \setminus A_{i}) \cap ( A_{j}) \neq \varnothing,$$
$$(W \setminus A_{i}) \cap (W \setminus A_{j}) \neq \varnothing.$$ But $A_{i} | (W \setminus A_i)$ and $A_{j} | (W \setminus A_{j})$ correspond to trivial splits in the start tree $T_{s}$ and are thus compatible with each other. Therefore, $A_{i} \cap A_{j}= \varnothing, |A_i\cap A_j| = 0$. So we have

$$\mid\bigcup_{i=1}^{|V_{\beta}|} A_{i}| \geq \sum _{i=1}^{|V_{\beta}|} |A_{i}| - \sum _{1\leq i < j \leq |V_{\beta}|} 0> \sum _{i=1}^{|V_{\beta}|} 1  = |V_{\beta}|.$$

This demonstrates that if $V_\alpha$ contains $d$ independent vertices from $V_{i}$, we need to remove at least $d+1$ vertices from $V_{\ell}$: $|V_{\alpha}| < d + n - d - 1 = n-1$, which contradicts our previous conclusion 
that $|V_{\alpha}| \geq n$.  Thus, $V_{\alpha}$ does not contain any vertex from $V_{i}$, $V_{\alpha} = V_{\ell}$, and $|V_\alpha| = n$. (See Figure~\ref{fig:BHVCG6} for examples.)

\end{proof}

\begin{thm}\label{thm:subgraph}Let $T_{s}$ be a start tree with $n$ leaves and at least one internal edge. Splitting one of the internal edges of $T_{s}$  results in two subtrees: $T_{a}$ with $a$ leaves and $T_{b}$ with $b$ leaves where\newline $a+b = n+2$. Then $G_{T_{a}, a, 1}$ and $G_{T_{b}, b, 1}$ are subgraphs of $G_{T_{s}, n, 1}$. Furthermore, $$|E(G_{T_{s}, n, 1})| = |E(G_{T_{a}, a, 1})| + |E(G_{T_{b}, b, 1})|+5(a-2)(b-2),$$ where $E(G)$ represents the edge set of a graph $G$.
\end{thm}
\begin{proof}
As shown in Figure~\ref{fig:concatenate}, we can obtain two subtrees by splitting one of the internal edges of tree $T_{s}$. If that edge defines the split ${P|Q}$, we will have two different representations for $T_{s}$. In Figure~\ref{fig:concatenate} (b), we use $A$ to represent the subtree with leaf vertices from $Q$; in Figure~\ref{fig:concatenate} (c), we use $B$ to represent the subtree with leaf vertices from $P$. If we treat both $A$ and $B$ as leaf vertices, we can view these two representations as two trees, $T_{a}$ and $T_{b}$, and construct $G_{T_{a}, |P|+1, 1}$ and $G_{T_{b}, |Q|+1, 1}$ by adding a new leaf $v$ to each tree. 

\begin{figure}[ht!]
\includegraphics[scale=.25]{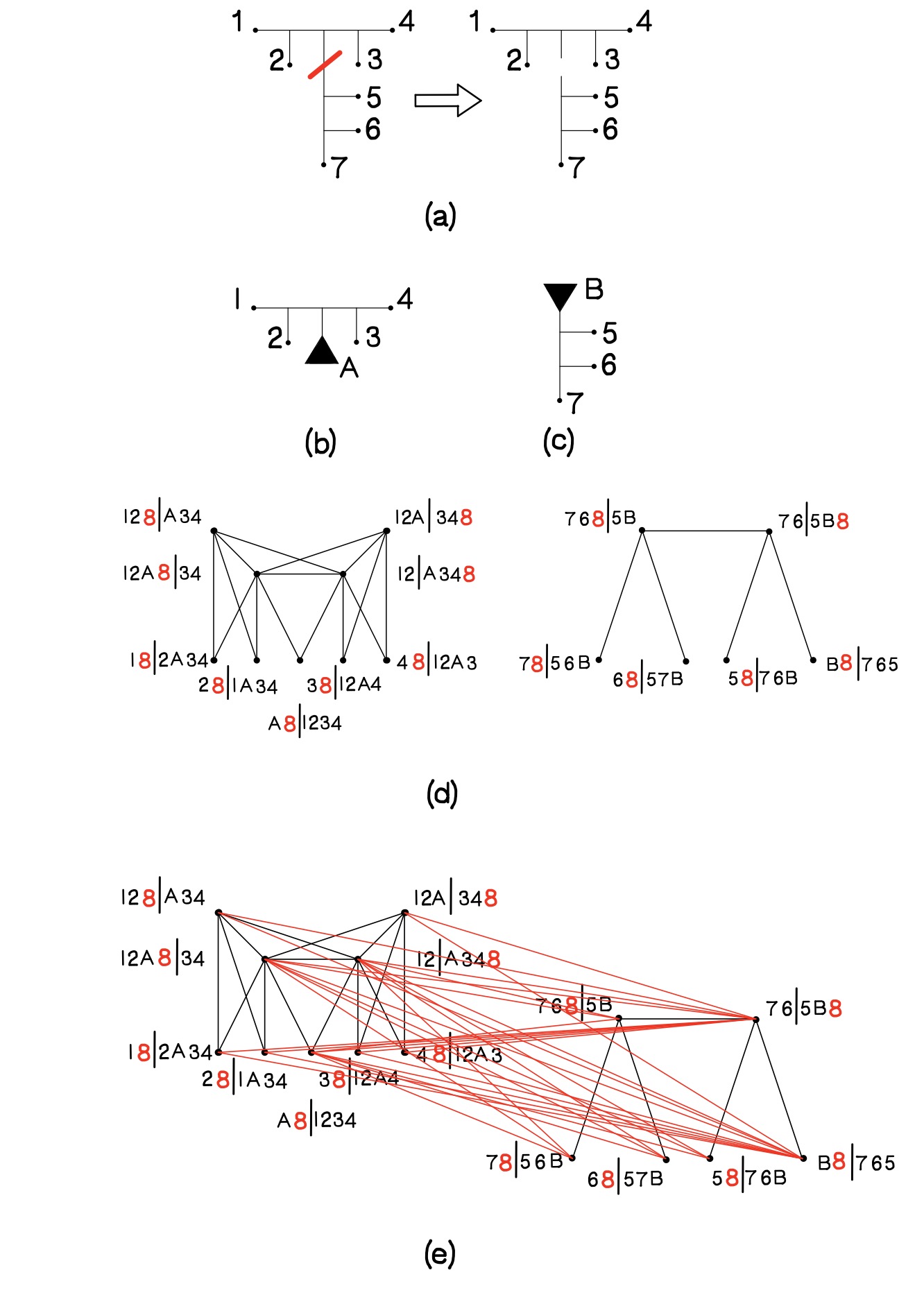}
\centering
\caption{Sub-Figure (a) shows an example for splitting an internal edge of a  tree $T_{s}$ to obtain two smaller trees. In Subfigures (b) and (c), the set $A$ corresponds to the subtree in $T_{s}$ with leaf vertex set $\{5, 6, 7\}$ and the set $B$ corresponds to the subtree in $T_{s}$ with leaf vertex set $\{1, 2, 3, 4\}$. Note figures (b) and (c) are two different representations of $T_{s}$.}
\label{fig:concatenate}
\end{figure}

When we identify $A$ and $B$ with the leaf vertex sets that they represent, the topologies of $G_{T_{a}, |P|+1, 1}$ and $G_{T_{b}, |Q|+1, 1}$ remains unchanged but the leaf vertices now have the same labels as the leaf vertices in $G_{T_{x}, |P|+|Q|, 1}$ obtained by adding $v$ to $T_{s}$. Note that the vertex set $(A\cup \{v\})|B$ and the vertex set $(B\cup \{v\})|A$ correspond to the two splits obtained by adding the new leaf $v$ to the edge ${P|Q}$. So by adding proper edges between $G_{T_{a}, |P|+1, 1}$ and $G_{T_{b}, |Q|+1, 1}$, we obtain $G_{T_{s}, |P|+|Q|, 1}$.

Using the definitions from Theorem~\ref{thm:independentsetinGwhenk=1}, we can define the vertex set: $V$, the internal edge split set: $V_{i}$, and the leaf edge split set: $V_{\ell}$ for $G_{T_{a}, |P|+1, 1}$, $G_{T_{b}, |Q|+1, 1}$, and $G_{T_{s}, |P|+|Q|, 1}$.  We denote these as $$V_{a}, V_{a_{i}}, V_{a_{\ell}}; V_{b}, V_{b_{i}}, V_{b_{\ell}};$$ and  $$V, V_{i}, V_{\ell}$$ respectively.  Observe that for all $u \in V_{b}$, if $u$ induces a split in the form ${M|N}$, then at least one of $M, N$ is a subset of $A \cup \{v\}$.

Next, we consider the new edges that will be added from leaf vertices in both graphs. Since leaf vertex sets will form  independent sets as shown in Theorem~\ref{thm:independentsetinGwhenk=1}, leaf vertices can only connect to internal vertices in the other tree-meaning leaf vertices added to $T_{b}$, for example, can only connect to internal edges adjacent to leaf vertices in $T_{a}$.

Leaf vertices in $V_{a}$ correspond to splits in the form of $\{u,v\}|B\setminus\{u\}, u\in B$. On the other hand, $B\cup\{v\}$ will be a subset of one side of the split for exactly half of the splits represented by internal vertices in $V_{b}$. There are $(a-1)(b-3)$ additional edges from the leaf vertices in $V_{a}$ when added to the internal vertices in $V_{b}$; similarly, there are $(b-1)(a-3)$ additional edges from the leaf vertices in $V_{b}$ added to the internal vertices in $V_{a}$. So another $(a-1)(b-3)+(b-1)(a-3)$ edges will be added in total.

The last case is the edges between $V_{a_{i}}$ and $V_{b_{i}}$. Half of the vertices in $V_{a_{i}}$ will represent splits with one side containing $A\cap \{v\}$ while for a split corresponding to a vertex in $V_{b_{i}}$, the side of the split that does not contain $B$ will be a subset of $A \cap \{v\}$. So half of the vertices in $V_{a_{i}}$ can be connected to all vertices in $V_{b_{i}}$. For the other half of vertices in $V_{a_{i}}$ that represent splits with $A$ and $v$ on different sides, the side with $v$ will be a subset of the side of a split represented by a vertex in $V_{b_{i}}$ that contains $B \cap \{v\}$, which is half of the vertices in $V_{b_{i}}$. Thus, $(a-3)(2(b-3))+(a-3)(b-3) = 3(a-3)(b-3)$ edges will be added. 

In total, the number of edges that need to be added to the $\BHV$ Connection Graph is
\begin{eqnarray*}
&3(a+b)& - 13 + (a-1)(b-3)+(a-3)(b-1)+3(a-3)(b-3)\\
&=& 5ab-10a-10b+20\\
&=&5(a-2)(b-2).
\end{eqnarray*}

\end{proof}

To illustrate the proof of Theorem \ref{thm:subgraph}, consider the situation illustrated in Figure \ref{fig:concatenate}: $A$ and $B$ contain only leaf vertices, Subfigures (b) and (c) are also two subgraphs of $T_{s}$, which we denote $T_{a}$ and $T_{b}$. From this perspective, we can construct the $\BHV$ Connection Graphs $G_{T_{a}, 5, 1}$ and $G_{T_{b}, 4, 1}$ with a new leaf with label 8, as shown in Subfigure (d). If we expand $A$ and $B$ to $\{5, 6, 7\}$ and $\{1, 2, 3, 4\}$ in the vertex names in $G_{T_{a}, 5, 1}$ and $G_{T_{b}, 4, 1}$, the union of the vertex sets in $G_{T_{a}, 5, 1}$ and $G_{T_{b}, 4, 1}$ will be equal to the vertex set of the $\BHV$ Connection Graph $G_{T_{s}, 7, 1}$ with the new leaf with label 8. Adding edges specified by \ref{thm:subgraph}, we will obtain the final $G_{T_{s}, 7, 1}$ shown in Subfigure (a).

\begin{figure}[ht!]
\centering
\label{fig:concatenateBHV}
\end{figure}

The following is one application of Theorem~\ref{thm:subgraph}.
\begin{cor}
The $\BHV$ Connection Graph $G_{T_{s}, n, 1}$ has $\frac{5}{2}(n-2)(n-3)$ edges.
\end{cor}
\begin{proof}
Any tree with more than three leaves can be formed by concatenating a three-leaf tree in the way described in Theorem~\ref{thm:subgraph} via iteration. See Figure~\ref{fig:iterativeTree} for an example. Using the last equation from Theorem~\ref{thm:subgraph}, the number of edges in $\BHV$ Connection Graph $G_{T_{s}, n, 1}$ is 
\[E(n) = E(n-1)+E(3)+5(n-1-2)(3-2) = E(n-1)+5(n-3),
\]
because $E(3) = 0.$

By solving the recurrence relation, we get the following explicit formula:
\[E(n) = \frac{5}{2}(n-2)(n-3).\]

\begin{figure}[ht!]
\includegraphics[scale=.15]{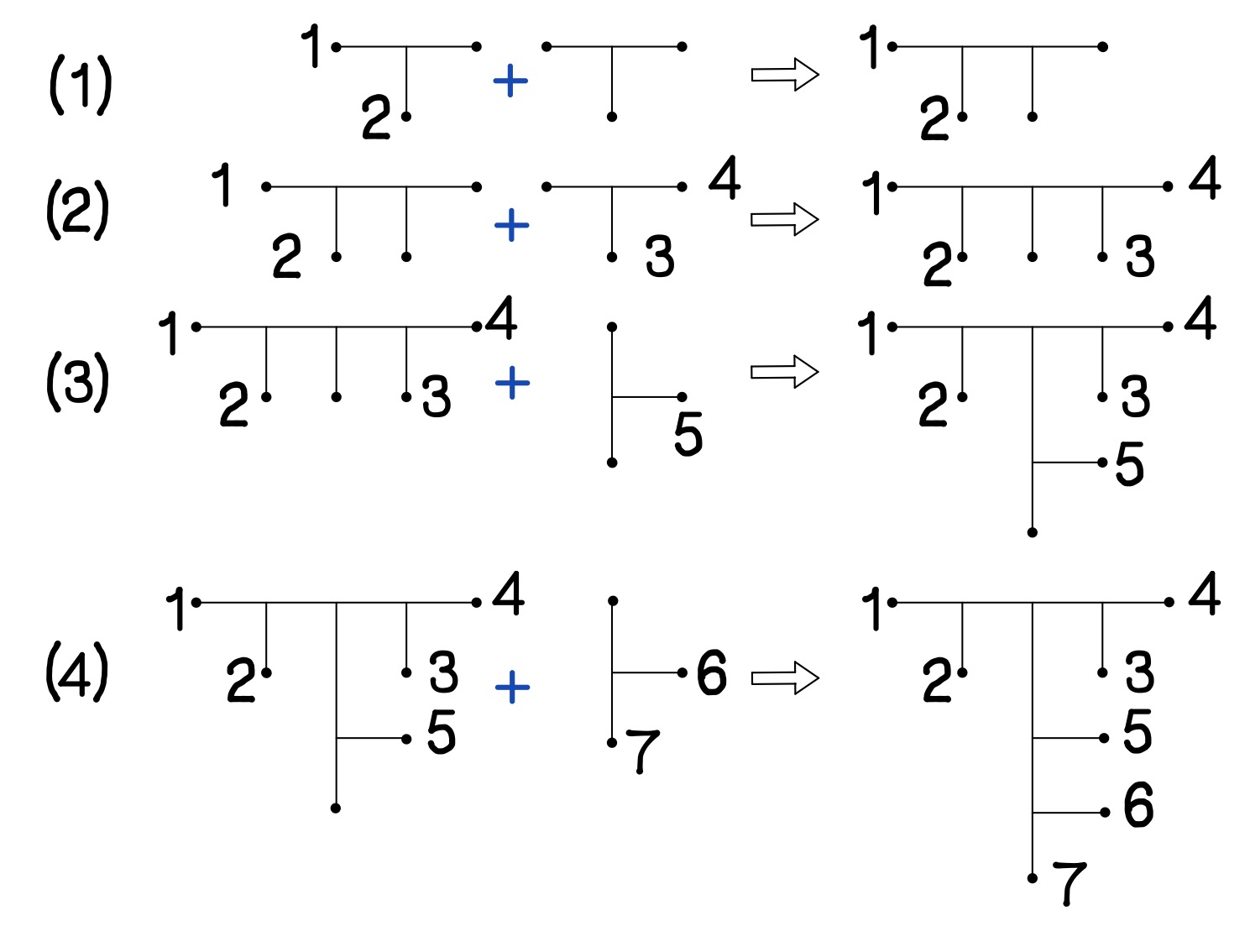}
\centering
\caption{The four iterative steps of adding a three-leaf tree to obtain a target tree $T_{t}$}
\centering
\label{fig:iterativeTree}
\end{figure}
\end{proof}

\section{Discussion}\label{sec:discussion}

The purpose of this manuscript is to provide a combinatorial method to transition between copies of $\BHV_{n}$ where $n$ is allowed to vary. Our combinatorial method only uses the intrinsic geodesic metric of \citep{owen2011fast} and the construction of the $\BHV$ Connection Space, rather than introducing a complex system of mathematical objects that make our results applicable only in limited situations-meaning situations where there are limits on the number of taxa, types of data (meaning conspecific data, gene data, or species trees that should be combined), or the shape of the input and output trees.  In particular, the types of tree topologies in the start tree $T_{s}$ in the beginning of the transition and the end of the the transition $T_{t}$ are not limited beyond the constraint that they are binary.  We make these comments to emphasize that our combinatorial method opens the door to the study of several problems in computational biology that up to this point were impossible to study in $\BHV$ spaces. We list a few examples below. 

%\textcolor{red}{This section needs to be revised after the github}

\begin{enumerate}

\item There is potential for new supertree construction methods (that widely vary and are still under development) for combining phylogenies on varying numbers of taxa with myriad biological properties \citep{wilkinson2005shape, bininda2004phylogenetic, akanni2015implementing}.  Supertree methods \emph{must} take inputs with varying numbers of taxa to be useful in a biological context. But in \citep{st2017shape} it was made clear that $\BHV$ spaces are well-suited for these methods due to the properties of the intrinsic geodesic metric. 

\item There is potential for new summary (coalesent-based) methods such as those developed in \citep{mirarab2014astral,liu2015coalescent} that combine gene phylogenies such as in \citep{joyner2014use} inferred from short samples from long genomes into a species phylogeny.  In practice, it is rare that genomic information for all species taxa is available, and complex methods for dealing with this issue are of current interest to the computational biology community \citep{streicher2015should, darriba2016prediction, baca2017ultraconserved} when working with biological data.  Summary methods are controversial \citep{chou2015comparative, springer2016gene}, but remain of deep interest to the computational biology community and always require exceptions and non-trivial software advances for cases with missing taxa that appeal to complex solutions \citep{xi2015impact}. While new methods for dealing with missing taxa are in constant production using computational techniques \citep{kobert2016computing, thomas2013pastis, molloy2017include}, it would be surprising if these results did not open the door to the use of $\BHV$ spaces in novel quantitative paradigms for the development of fast and accurate novel methods for summary-based species tree estimation. 

\item Consensus methods such as the majority consensus method can achieve better results when there is not a restriction on the input trees having the same number of taxa. As pointed out in \citep{st2017shape} the geodesic metric allows the construction of paths between trees in a set of trees $\mathcal{T}$ in $\BHV$ spaces that do not inherit splits from trees not in the set $\mathcal{T}$.  Therefore our results may enable a novel quantitative embedding of the consensus problem for phylogenies that is competitive with our outperforms pre-existing methods. We believe this is a natural extension of this project because consensus methods rely on splits, but $\BHV$ spaces have geometric components that allow for polytomic trees, which are not informative in the use of consensus methods that rely on resolved internal splits in a phylogeny. 

\item There is also potential for optimization of the results in this paper by further study of the $\BHV$ Connection Space and the $\BHV$ Connection Cluster.  These are novel objects that may provide a more accessible quantitative framework for addressing problems such as quartet-agglomeration \citep{sumner2017developing, sayyari2016anchoring, reaz2014accurate, avni2015weighted, rusinko2012invariant} that require combining four-taxon trees into trees with any possible number of taxa. We mention the example of quartet-agglomeration in particular because, as explained thoroughly in \citep{sumner2017developing}, the Neighbor-Joining method of tree inference \citep{saitou1987neighbor} is fundamentally a quartet-based method that continues to perform well on many datasets that contain more than four taxa \citep{yoshida2016efficiencies}.  The scientific connection between quartet-agglomeration and Neighbor-Joining should not be ignored if progress in these areas is to be made, and our results provide a way to study these problems in a new setting. 

\item Finally, there is potential for using our new quantitative framework in the study of inferred unrooted trees that contain polytomies (vertices of degree of four or higher). The orthants of non-maximal dimension in any $\BHV_{n}$-space correspond to polytomic trees. Polytomies are often recovered in studies of biological data due to the fact that the biological signal in any dataset may not be strong enough to indicate that branch length in an inferred tree should have length greater than zero. This is known to be an issue in maximum-likelihood tree estimation on biological datasets \citep{pamminger2017testing, simmons2014divergent, slowinski2001molecular}.  Further, it is known that if the true evolutionary history cannot provide sufficient information to resolve a polytomy, distance-based methods, which remain in use not only on their own, but also as components of maximum-likelihood phylogenomic software inference packages such as FastTree-2 \citep{price2010fasttree} are biased against returning the  correct tree \citep{davidson2014distance}.

\end{enumerate}

\section{Supporting materials} 

%\textcolor{red}{Add the figure from 9/28 here to show what the software does and explain here that the figures in the paper were drawn on a Ipad to the clarify the ideas}

To generate connection graphs connecting $\BHV$ spaces of specific dimensions using our software, use the script connection graph.py which includes dependencies on Dendropy (see  \citep{sukumaran2010dendropy} for the original paper about this software package). Go to \begin{verbatim}
    https://github.com/cpmoni/igl-polyhedra
\end{verbatim}

to find installation and usage instructions.  Figure~\ref{fig:connectiongraphsample} is the connection graph generated to connect $\BHV$ spaces between 9 and 10 dimensions.  Please note that the other graphs in this paper is drawn by hand rather than the software, in order to fully clarify the mathematical concepts.

\begin{figure}[ht!]
\includegraphics[scale=.6]{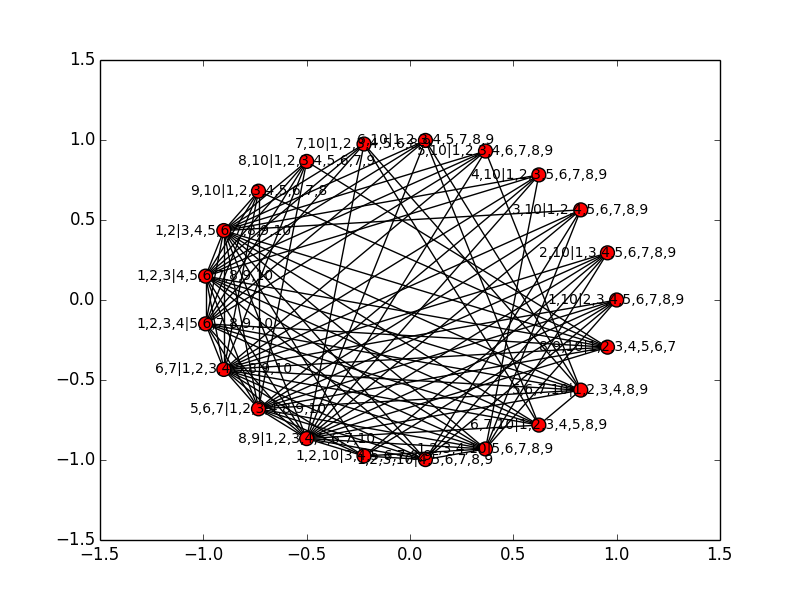}
\centering
\caption{Visualization of the $\BHV$ connection graph connecting 9 and 10 dimensions \\ 
generated by the code released on the github with this paper}
\centering
\label{fig:connectiongraphsample}
\end{figure}

\section{Acknowledgments}
The undergraduate students Y.R, S.Z, J.B., and J. S., as well as the initialization of this project, were supported by a Mathways Grant NSF DMS-1449269 to the Illinois Geometry Lab at the University of Illinois Urbana-Champaign. R.D. was supported by the NSF grant DMS-1401591. M.D. was supported by the NSF Graduate Research Fellowship DGE-1144245.  C.M. was supported by a GAANN fellowship from the Department of Education awarded by the University of Illinois Urbana-Champaign.

\bibliographystyle{kp}
\bibliography{refs}

\end{document}